\documentclass[a4paper,UKenglish,cleveref, autoref, thm-restate]{lipics-v2021}

\pdfoutput=1 
\hideLIPIcs  

\title{Quantum Query-Space Lower Bounds Using Branching Programs} 

\author{Debajyoti Bera}{Department of Computer Science, IIIT-Delhi, New Delhi, India. \and Center for Quantum Technologies, IIIT-Delhi, New Delhi, India.}{dbera@iiitd.ac.in}{}{}

\author{Tharrmashastha SAPV\footnote{Corresponding Author}}{Department of Computer Science, IIIT-Delhi, New Delhi, India.}{tharrmashasthav@iiitd.ac.in}{}{}

\authorrunning{D. Bera and Tharrmashastha SAPV} 

\Copyright{Debajyoti Bera and Tharrmashastha SAPV} 



\category{} 

\relatedversion{} 

\acknowledgements{The authors would like to thank Dr. Paul Beame for his insights into this work and for pointing out the effects of the results on low-space Turing machines.}

\nolinenumbers 

\EventEditors{John Q. Open and Joan R. Access}
\EventNoEds{2}
\EventLongTitle{42nd Conference on Very Important Topics (CVIT 2016)}
\EventShortTitle{CVIT 2016}
\EventAcronym{CVIT}
\EventYear{2016}
\EventDate{December 24--27, 2016}
\EventLocation{Little Whinging, United Kingdom}
\EventLogo{}
\SeriesVolume{42}
\ArticleNo{23}

\usepackage{physics}
\usepackage{xcolor}
\usepackage{xspace}
\usepackage{todonotes}
\usepackage{tikz}
\usepackage{graphicx}
\usepackage{float}
\usepackage{thmtools}
\usepackage{thm-restate}
\usepackage{algorithmicx, algorithm}
\usepackage{algpseudocode}

\newtheorem{problem}{Problem}

\newcommand{\nqbp}{NQBP\xspace}
\newcommand{\aqbp}{AQBP\xspace}

\newcommand{\gqbp}{GQBP\xspace}
\newcommand{\rgqbp}{r-GQBP\xspace}

\newcommand{\scrg}{$\mathcal{G}$\xspace}  
\newcommand{\expec}{\mathop{\mathbb{E}}}

\begin{document}

\maketitle

\begin{abstract}
Branching programs are quite popular for studying time-space lower bounds.
Bera et al.\ recently introduced the model of generalized quantum branching program {\it aka.} GQBP that generalized two earlier models of quantum branching programs. In this work we study a restricted version of GQBP with the motivation of proving bounds on the query-space requirement of quantum-query circuits.

We show the first explicit query-space lower bound for our restricted version. We prove that the well-studied OR$_n$ decision problem, given a promise that at most one position of an $n$-sized Boolean array is a 1, satisfies the bound $Q^2 s = \Omega(n^2)$, where $Q$ denotes the number of queries and $s$ denotes the width of the GQBP. 

We then generalize the problem to show that the same bound holds for deciding between two strings with a constant Hamming distance; this gives us query-space lower bounds on problems such as Parity and Majority. Our results produce an alternative proof of the $\Omega(\sqrt{n})$-lower bound on the query complexity of any non-constant symmetric Boolean function when the space is restricted to $\log n + O(1)$ qubits.
\end{abstract}

\newpage

\section{Introduction}
\label{sec:intro}

Can a quantum algorithm make fewer queries by cleverly using lots and lots of ancill\ae? Could it be that a speedup in the number of queries in comparison to classical algorithms (this is often quadratic) vanishes if only a few qubits are available? This is the essence of the query-space tradeoff that motivated this work. 

These questions are central to understanding the resource requirements of computational problems and, more importantly, the relationship between them. These have been studied in the form of time-space tradeoffs for classical algorithms. One of the favoured lines of attack for doing so is via analysis of branching programs.

A (classical) branching program is a generalization of a decision tree; structurally, a branching program looks like a dag in which the nodes are labelled using input variables. On any input, a branching program starts at the root node and evolves by querying the variable corresponding to the current node and then following the edge corresponding to the value of that variable. The input is accepted or rejected if the final node reached is marked accept or reject. This model is attractive since it nicely characterizes the space and time required by a Turing machine. The length (of the longest path in the dag corresponding to an input) can be related to the number of steps in a Turing machine taken on that input. The width, which is the largest number of nodes in any level of the dag, can also be related to the space complexity of a computation. Thus, branching programs are traditionally used to prove lower bounds on space and the space-time tradeoffs of various problems.

The study of space-time tradeoffs in the classical setting has been a long-standing one.
Despite numerous works on classical time-space trade-off complexity in various settings (\cite{abelson1979note, grigor1976application, tompa1978time, savage1978space, reischuk1980improved, pippenger1978time, munro1980selection, lengauer1979upper}), Borodin et al., were the first to study time-space lower bounds using branching programs, but one that allowed only comparison-based queries (\cite{borodin1979time}).
They showed that any classical algorithm for sorting $N$ inputs using comparison queries requires the time-space product $TS = \Omega(N^2)$ where $T$ and $S$ are the time and space complexities of the branching program, respectively.
Later in \cite{borodin1980time}, Boroding et al. showed that any $R$-way branching program for sorting $N$ integers with $R=N^2$ should satisfy $TS = \Omega(N^2/\log N)$.
This was followed by several time-space lower-bound results (\cite{yesha1984time,abrahamson1986time}) that used branching programs for various matrix problems.
Finally, Beame~\cite{beame1989general} showed that any $R$-way BP for sorting with inputs in the range $[1, R]$ and $R\ge N$ requires $TS = \Omega(N^2)$, which is optimal up to logarithmic factors.

The first result concerning classical query-space trade-off lower bounds for a \textit{decision problem} was proved in \cite{borodin1987time} for the element distinctness problem.
They showed that any $T$-query $S$-space deterministic comparison branching program that decides element distinctness should satisfy $TS = \Omega(n^{3/2}\sqrt{\log n})$.
Yao, in \cite{yao1988near}, improved this to $TS=\Omega(n^{2-\varepsilon(n)}), \varepsilon(n)=O(1/(\log n)^{1/2})$.
While showing any non-trivial time-space tradeoff lower bound for a natural problem in general branching programs was hard, Beame et al. (\cite{beame1998time}) constructed an explicit function to obtain the first non-trivial time-space trade-off lower bound in general branching programs.
They also showed a non-trivial time-space trade-off lower bound on the R-way branching program for another class of explicitly constructed functions.
This was followed by the work of Ajtai (\cite{ajtai2002determinism}) in which, along with other results, they give lower bounds for non-probabilistic algorithms for the element distinctness problem.
They prove that any two-sided error BP of subexponential size must have length $T = \Omega\bigg(n\bigg(\frac{\log\log n}{\log\log\log n}\bigg)\bigg).$
Using techniques that were essentially extensions of that in \cite{ajtai2002determinism} and \cite{beame1998time}, Beame et al. (\cite{beame2003time}) proved that any algorithm for the randomized computation of the element distinctness problem using space $S$ should satisfy $T = \Omega\bigg(n\sqrt{\log(n/S)/\log\log(S)}\bigg).$

In the quantum domain, some works prove time-space lower bounds for problems like sorting(\cite{klauck2007quantum}), $k$-disjoint collisions problems(\cite{hamoudi2020quantum}), Boolean matrix-vector multiplication(\cite{klauck2007quantum, beame2024quantum}), matrix-matrix multiplication(\cite{klauck2007quantum, beame2024quantum}), solving a system of inequalities(\cite{ambainis2006new}) and multiple matrix problems(\cite{beame2024quantum}).
However, to the best of our knowledge, no non-trivial time-space lower bound is known for any decision problem in any quantum model.

A quantum branching program would be a natural tool to derive space and space-time lower bounds. There are three types of quantum branching programs to the best of our knowledge. All the models maintain a global state (which could be a superposition of well-defined states) that changes upon querying some input bit. Nakanishi et al.~\cite{nakanishi2000ordered} used their model (that we refer to as NQBP following the nomenclature used by Bera et al.~\cite{bera2023generalized}) to separate ordered bounded-width NQBP and ordered bounded-width probabilistic branching programs (observe that a quantum branching program is by nature probabilistic). Ablayev et al.~\cite{ablayev2001computational} defined another model (that we refer to as AQBP), which they used to prove that $NC^1$ (the class of $O(\log n)$-depth and polynomial-size Boolean circuits) is contained in width-2 AQBPs. Very recently, Bera et al.~\cite{bera2023generalized} designed another model (that we refer to as GQBP) that generalized both the earlier models. They proved that problems in $NC^1$ can be solved by width-2 GQBPs, and further, any symmetric Boolean function can be computed by width-2 length-$O(n^2)$ GQBP.

Since the GQBP generalizes the two other models, we chose that model to derive our lower bounds.

\section{Overview of Results}

We summarise our contributions in this section.\\

\noindent{\bf 1. r-GQBP:} Our motivation was to derive time-space lower bounds of well-studied computational problems using the length and width dependencies of GQBPs. However, the general nature of the structure and the transition function of a GQBP made it challenging to bound.

Hence, we first define a restricted form of GQBP that we call r-GQBP. Intuitively, the paths (in superposition) followed for every input are identical, the only difference being the weights, rather the amplitudes, on the edges. Despite the rigid structure, we could prove that r-GQBPs can simulate any quantum query circuit and vice versa as efficiently as (unrestricted) GQBPs.\\

\noindent{\bf 2. Lower-bound for OR:}
The OR problem is one of the typical problems when it comes to proving a lower bound; given an $n$-bit Boolean array, the problem is to determine if there is any 1 in the array. We prove a trade-off between the length and the width of any r-GQBP, solving a promise version of the OR problem in which the array has either one 1 or no ones.

\begin{restatable}{theorem}{lowerboundor}
    \label{thm:lowerbound-or}
    Any r-GQBP of length $L$ and width $s$ that solves the Promise-OR problem has to satisfy $L = \Omega\bigg(\frac{n}{\sqrt{s}}\bigg)$.
\end{restatable}

\noindent{\bf 3. Lower-bound for Parity and Majority:} The OR problem is a special case of the $(k,k+\delta)$-Hamming decision problem, which is about deciding whether a binary string has $k$ ones or $k+\delta$ ones, given the promise that one of them must hold. Next, we prove an identical length-width lower bound on the Hamming decision problem, too. This gives us identical bounds for functions such as Parity, AND, Majority, etc.

\begin{restatable}{theorem}{hammingdecision} 
    \label{thm:lowerbound-khamprob}
    Any r-GQBP of length $L$ and width $s$ that solves the ($k,k+\delta$) Hamming Decision problem has to satisfy $L = \Omega\bigg(\frac{n}{\delta\sqrt{s}}\bigg)$.
\end{restatable}

\noindent{\bf 4. Lower-bound for non-constant symmetric function:}
Using the above results and an observation about non-constant symmetric functions, we are able to prove the following trade-off.

\begin{restatable}{theorem}{lowerboundsymmgqbp}\label{thm:lowerbound-gqbp-symm-function}
    Any ($s,L$)-\rgqbp that computes a non-constant symmetric function $f:\{0,1\}^n \rightarrow \{0,1\}$ has to satisfy $L = \Omega\bigg(\frac{n}{\sqrt{s}}\bigg)$.
\end{restatable}

Beals et al.~\cite{beals2001quantum} introduced the polynomial method for quantum lower bounds; they showed that the number of queries, denoted $Q(f)$, to compute any non-constant symmetric function $n$-bit function $f$ is $\Omega(\sqrt{n})$~\footnote{Technically, their bound is a bit stronger. But their bound collapses to this one for functions that flip their value on inputs with constant Hamming weights or Hamming weights close to $n$.}, and this is without any restriction on the number of qubits. Using the above theorem and our quantum circuit to \gqbp conversion, we could prove the same lower bound with a restriction on space usage.

\begin{restatable}{theorem}{symmetricfunction}
    \label{corr:circ-lowerbound}
    Any quantum query circuit that computes a non-constant symmetric function $f:\{0,1\}^n \rightarrow \{0,1\}$ with $\log n + O(1)$ qubits has to satisfy $Q(f) = \Omega(\sqrt{n})$.
\end{restatable}

In the classical setting, branching programs are a reasonable model of computation to study due to their relations with circuits and Turing machines. 
We believe the same is true for the quantum branching programs.
While the query-space lower bounds for the query circuits derived from that of the r-GQBPs may not be significant, we believe these lower bounds could be used to derive significant lower bounds for quantum circuits and Turing machines. 
Although a more hardcore analysis is required, the lower bound results on \rgqbp informally imply that it is not possible to obtain a quadratic speedup using sub-logarithmic space in Turing machines when computing arbitrary symmetric functions.

\section{Background: Generalized Quantum Branching Programs}
\label{sec:background}

There are a few different models of quantum branching programs all of which try to bring some {\em quantum feature} to the classical model. The first model was introduced by Nakanishi et al.\ that allows the state transitions to follow multiple paths (in superposition)~\cite{nakanishi2000ordered}; however, to decide acceptance and rejection, a measurement is performed after every transition. They were able to show that there exists an ordered bounded-width \nqbp that computes the ``Half Hamming'' weight function --- a function that indicates if the Hamming weight of the input is $n/2$, but no ordered bounded-width probabilistic branching program can compute the same.

Later, Ablayev et al.\ introduced another quantum branching program model~\cite{ablayev2001computational}.
They showed that while any ordered binary decision diagram has to be of a width at least $p_n$ to compute the $Mod_{p_n}$ function, which indicates if the Hamming weight of the input is divisible by the prime $p_n$, there is an $O(\log(p_n))$ width \aqbp that computes the $Mod_{p_n}$ function with one-sided error.
In addition, they showed that any language in the $NC1$ complexity class is exactly decidable by width-2 AQBPs of polynomial length.

Sauerhoff et al.\ presented a few interesting results concerning a restricted version of NQBPs~\cite{sauerhoff2005quantum}.
Recently, Maslov et al showed that any symmetric Boolean function can be computed by some \aqbp of width $2$ and length $O(n^2)$~\cite{maslov2021quantum}.

Generalized quantum branching programs were recently introduced by Bera et. al~\cite{bera2023generalized}.

\begin{definition}[Generalized quantum branching program (\gqbp)]
    \label{def:gqbp}
    A generalized quantum branching program is a tuple $(Q, E, \ket{v_0}, L, \delta,  F)$ where
    \begin{itemize}
        \item $Q = \bigcup_{i=0}^{l} Q_i$ is a set of nodes where $Q_i$ is the set of nodes at level $i$.
        \item $E$ is a set of edges between the nodes in $Q$ such that $(Q,E)$ is a levelled directed acyclic multi-graph.
        \item $\ket{v_0}$ is the initial state which is a superposition of the nodes in $Q_0$.
        \item $\delta : Q\times \{0,1\} \times Q \xrightarrow{} \mathbb{C}$ is a quantumly well-behaved transition function.
        \item $L : Q \xrightarrow{} \{1,2\cdots, n\}$ is a function that assigns a variable to each node in $Q\setminus F$.
        \item $F\subseteq Q$ is the set of terminal nodes.
    \end{itemize}
\end{definition}

Let's understand the evolution of \gqbp since our results are in this model. Let $x$ denote the input string. Consider any node $v$ at some level; $v$ is labelled with some input variable, say $x_i$, that can be derived as $i=L(v)$. There are two sets of outgoing transitions from $v$, one corresponding to $x_i=0$ and another corresponding to $x_i=1$, and each of them is to one or more nodes in the next level. Furthermore, each transition, say $v \substack{0 \\ \mapsto} w$ is associated with a complex $\delta(v,0,w)$; the intuition being that, if $x_i=0$, the state will be changed to a superposition of all the target nodes of transitions corresponding to 0, and the amplitudes of the superpositions will be specified by the corresponding $\delta$ values. The notion of ``quantum well-behavedness'' essentially ensures that the superposition is a quantum state. We can state the process more formally as follows.

\begin{enumerate}
    \item Given access to an input $x$ using oracle $O_x$, define $U^{O_x}_i : \mathcal{H}(d)\xrightarrow{}\mathcal{H}(d)$ for $i\in \{1,2,\cdots, l\}$ that has query access to $O_x$ as
    \begin{equation*}
        U^{O_x}_i \ket{v} =  \sum_{v'\in Q_i}\delta(v, O_x(L(v)), v')\ket{v'}
    \end{equation*}
    for any $\ket{v}\in Q_{i-1}$ where $\mathcal{H}(d)$ is the Hilbert space on $d$ dimensions with $d$ as the width of the QBP.
    \item  Start at $\ket{v_0}$.
    \item For $i=1$ to $l$, perform $\ket{\psi_{i}} = U^{O_x}_{i}\ket{\psi_{i-1}}$
    where $\ket{\psi_0} = \ket{v_0}$.
    \item Measure $\ket{\psi_{l}}$ in the standard basis. If the observe node is in $F$ then the input is accepted, else rejected.
\end{enumerate}

We say that a QBP \scrg \textbf{exactly computes} a function $f$ if for all inputs $x$, \scrg accepts $x$ if $f(x)=1$ and rejects $x$ if $f(x)=0$, both with probability $1$.
Similarly, a QBP \scrg is said to \textbf{bounded-error compute} a $f$ if for all inputs $x$, \scrg accepts $x$ if $f(x)=1$ and rejects $x$ if $f(x)=0$, both with probability at least $2/3$. We stick to the latter in this work; wherever we talk about a \gqbp\ {\em solving} some function, we mean that in the bounded-error sense.

There are two primary complexity measures of a \gqbp. The \textit{length} of a \gqbp is the length of the longest path from the root node to any single node.
The \textit{width} of a \gqbp is the maximum number of nodes on any level of the \gqbp.
A \gqbp $P$ is said to be a ($w,l$)-\gqbp if $P$ is a $w$ width $l$ length \gqbp.

In addition to introducing the \gqbp model, Bera et al. demonstrated that the \gqbp model generalizes the \nqbp and the \aqbp models. More specifically, they showed that any ($w,l$)-\aqbp is also a ($w,l$)-\gqbp and that any $s$-sized $l$-length \nqbp can be simulated exactly by a ($s,l$)-\gqbp. However, neither \gqbp nor the other models have been used to prove length or width lower bounds on any explicit function.

\section{Restricted-GQBP}
\label{sec:rgqbp}

In this section, we introduce a restricted version of the \gqbp model called the restricted-\gqbp (or, \rgqbp) and show some results concerning them.
More importantly, we show that the \rgqbp model is equivalent to the quantum query circuit model.
This equivalence motivates the investigation of query-space lower bounds on the r-GQBPs, which could be translated to query-space lower bounds on quantum query circuits.


\begin{definition}
    A GQBP is called a restricted-GQBP (or r-GQBP) if, for any node $m$ on any level, its transition vectors $v^0_m$ and $v^1_m$, corresponding to the $0$ and $1$ transitions, differ only by a phase, i.e., there exists some $\theta$ such that $v^1_m = e^{i\theta}v^0_m$.
\end{definition}

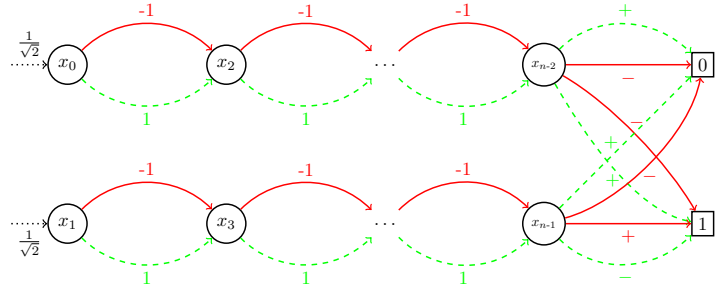
\begin{figure}[ht!]

\begin{center}
\scalebox{0.7}{
    \begin{tikzpicture}[node distance={30mm}, thick, main/.style = {draw, circle}, empty/.style = {circle}, rect/.style={draw, rectangle}] 
            \node[empty] (0) {$\hdots$};
            \node[empty] (1) [below of=0] {$\hdots$};
            \node[main] (2) [left of=0] {$x_2$};
            \node[main] (3) [left of=1] {$x_3$};
            \node[main] (4) [left of=2]{$x_0$}; 
            \node[main] (5) [left of=3] {$x_1$};
            \node[main] (6) [right of=0] {\scalebox{0.75}{$x_{n\text{-}2}$}};
            \node[main] (7) [right of=1] {\scalebox{0.75}{$x_{n\text{-}1}$}};
            \node[rect] (8) [right of=6] {$0$};
            \node[rect] (9) [right of=7] {$1$};
            \node[empty](10) [left of=4, xshift=50pt] {};
            \node[empty](11) [left of=5, xshift=50pt] {};

            \draw[->, dotted] (10) to node[above] {$\frac{1}{\sqrt{2}}$} (4);
            \draw[->, dotted] (11) to node[below] {$\frac{1}{\sqrt{2}}$} (5);
            
            \draw[->, green, dashed] (4) to [out=-45,in=-135,looseness=1] node[below] {1} (2);
            \draw[->, red] (4) to [out=45, in=135,looseness=1] node[above] {-1} (2);
            \draw[->, green, dashed] (5) to [out=-45,in=-135,looseness=1] node[below] {1} (3);
            \draw[->, red] (5) to [out=45, in=135,looseness=1] node[above] {-1} (3);

            \draw[->, green, dashed] (2) to [out=-45,in=-135,looseness=1] node[below] {1} (0);
            \draw[->, red] (2) to [out=45, in=135,looseness=1] node[above] {-1} (0);
            \draw[->, green, dashed] (3) to [out=-45,in=-135,looseness=1] node[below] {1} (1);
            \draw[->, red] (3) to [out=45, in=135,looseness=1] node[above] {-1} (1);

            \draw[->, green, dashed] (0) to [out=-45,in=-135,looseness=1] node[below] {1} (6);
            \draw[->, red] (0) to [out=45, in=135,looseness=1] node[above] {-1} (6);
            \draw[->, green, dashed] (1) to [out=-45,in=-135,looseness=1] node[below] {1} (7);
            \draw[->, red] (1) to [out=45, in=135,looseness=1] node[above] {-1} (7);

            \draw[->, green, dashed] (6) to [out=45,in=135,looseness=1] node[above] {+} (8);
            \draw[->, red] (6) to node[below] {$-$} (8);
            \draw[->, green, dashed] (6) to [out=-60,in=170,looseness=0.75] node[below] {+} (9);
            \draw[->, red] (6) to [out=-30,in=120,looseness=0.75] node[above] {$-$} (9);
            
            \draw[->, green, dashed] (7) to [out=45,in=-135,looseness=0.75] node[left] {+} (8);
            \draw[->, red] (7) to [out=15,in=-100,looseness=0.75] node[below] {$-$} (8);
            \draw[->, green, dashed] (7) to [out=-45,in=-135,looseness=1] node[below] {$-$} (9);
            \draw[->, red] (7) to node[below] {+} (9);

            
        \end{tikzpicture} 
        }
        \end{center}
        \caption{The \gqbp presented in~\cite{bera2023generalized} for the $Parity_n$ problem (reproduced here with the authors' kind permission) was found to be an \rgqbp. $+$ denotes $\tfrac{1}{\sqrt{2}}$ and $-$ denotes $-\tfrac{1}{\sqrt{2}}$. The green (dotted) and the red (dashed) lines correspond to the $0$ and the $1$ transitions. \label{fig:gqbp-parity}}
    \end{figure}

    

This restriction on the transition vectors of each node presents a nice structure to the \rgqbp model, providing the following lemma as a consequence.
Later, we exploit this structure to prove query-space lower bounds on \rgqbp and even for quantum query circuits. See \cref{fig:rgqbp-split} for an illustration of the next lemma.

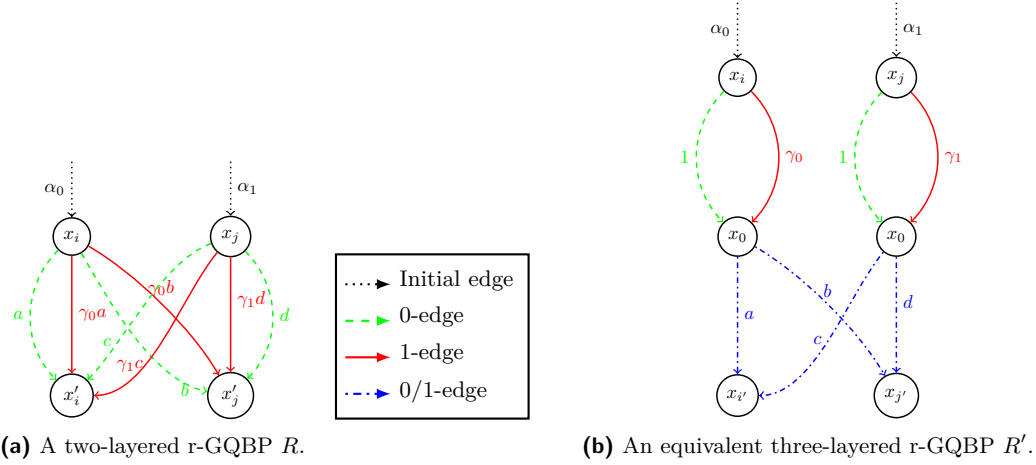
\begin{figure}
    \centering
    \begin{subfigure}{0.5\textwidth}
        \centering
\label{fig:rgqbp}
\scalebox{0.7}{
        \begin{tikzpicture}[node distance={30mm}, thick, main/.style = {draw, circle}, empty/.style = {circle}, rect/.style={draw, rectangle}] 
            \node[main] (0) {$x_i$}; 
            \node[main] (1) [right of=0] {$x_j$};
            \node[main] (2) [below of=0] {$x_i'$};
            \node[main] (3) [right of=2]{$x_j'$}; 
            \node[empty](4) [above of=0, yshift=-40pt] {};
            \node[empty](5) [above of=1, yshift=-40pt] {};

            \draw[->, dotted] (4) to node[left] {$\alpha_0$} (0);
            \draw[->, dotted] (5) to node[right] {$\alpha_1$} (1);
            
            \draw[->, green, dashed] (0) to [out=-135,in=135,looseness=1] node[left] {$a$} (2);
            \draw[->, red] (0) to node[right] {$\gamma_0 a$} (2);
            \draw[->, green, dashed] (0) to [out=-60,in=170,looseness=0.75] node[below right=30pt] {$b$} (3);
            \draw[->, red] (0) to [out=-30,in=120,looseness=0.75] node[above] {$\gamma_0 b$} (3);

            \draw[->, green, dashed] (1) to [out=200,in=45,looseness=0.75] node[below left=20pt] {$c$} (2);
            \draw[->, red] (1) to [out=230,in=0,looseness=0.75] node[below left=10pt] {$\gamma_1 c$} (2);
            \draw[->, green, dashed] (1) to [out=-45,in=45,looseness=1] node[right] {$d$} (3);
            \draw[->, red] (1) to node[above right] {$\gamma_1 d$} (3);
\end{tikzpicture} \hspace*{1em}
}  \begin{tikzpicture}[node distance={30mm}, thick, main/.style = {draw, circle}, empty/.style = {circle}, rect/.style={draw, rectangle}]
        \footnotesize
            \matrix [draw,below right] at (current bounding box.north east) {
                \draw[-latex,color=black, dotted](0,-0.3) -- (0.6,-0.3); & \node{Initial edge};\\
                \draw[-latex,color=green, dashed](0,-0.3) -- (0.6,-0.3); & \node{$0$-edge};\\
                \draw[-latex,color=red](0,-0.3) -- (0.6,-0.3); & \node{$1$-edge};\\
                \draw[-latex,color=blue, dash pattern = on 1 pt off 2 pt on 3 pt off 2 pt](0,-0.3) -- (0.6,-0.3); & \node{$0/1$-edge};\\
            };
            
        \end{tikzpicture} 
        \caption{A two-layered \rgqbp $R$.}
    \end{subfigure}
    \hfill
    \begin{subfigure}{0.45\textwidth}
        \centering
        \label{fig:rgqbp-two-layer}
        \scalebox{0.7}{
        \begin{tikzpicture}[node distance={30mm}, thick, main/.style = {draw, circle}, empty/.style = {circle}, rect/.style={draw, rectangle}] 
            \node[main] (0) {$x_i$}; 
            \node[main] (1) [right of=0] {$x_j$};
            \node[main] (2) [below of=0] {$x_0$};
            \node[main] (3) [right of=2]{$x_0$}; 
            \node[empty](6) [above of=0, yshift=-40pt] {};
            \node[empty](7) [above of=1, yshift=-40pt] {};
            
            \node[main] (4) [below of=2] {$x_{i'}$};
            \node[main] (5) [below of=3] {$x_{j'}$};

            \draw[->, dotted] (6) to node[left] {$\alpha_0$} (0);
            \draw[->, dotted] (7) to node[right] {$\alpha_1$} (1);
            
            \draw[->, green, dashed] (0) to [out=-135,in=135,looseness=1] node[left] {$1$} (2);
            \draw[->, red] (0) to [out=-45,in=45,looseness=1] node[right] {$\gamma_0$} (2);

            \draw[->, green, dashed] (1) to [out=-135,in=135,looseness=1] node[left] {$1$} (3);
            \draw[->, red] (1) to [out=-45,in=45,looseness=1] node[right] {$\gamma_1$} (3);

            \draw[->, blue, dash pattern = on 1 pt off 2 pt on 3 pt off 2 pt] (2) to node[right] {$a$} (4);
            \draw[->, blue, dash pattern = on 1 pt off 2 pt on 3 pt off 2 pt] (2) to [out=-30,in=120,looseness=0.75] node[above] {$b$} (5);

            \draw[->, blue, dash pattern = on 1 pt off 2 pt on 3 pt off 2 pt] (3) to [out=230,in=0,looseness=0.75] node[left] {$c$} (4);
            \draw[->, blue, dash pattern = on 1 pt off 2 pt on 3 pt off 2 pt] (3) to node[above right] {$d$} (5);
            
        \end{tikzpicture} 
        }
        \caption{An equivalent three-layered \rgqbp $R'$.}
    \end{subfigure}
    \caption{Figure (a) shows a two-layered one-transition \rgqbp $R$. Figure (b) is a three-layered two-transition \rgqbp $R'$ that is equivalent to $R$. Note that the first transition of $R'$ is query-dependent, whereas the second transition of $R'$ is query-independent.}
    \label{fig:rgqbp-split}
\end{figure}

\begin{lemma}
    \label{lemma:rgqbp-layer-split}
    The transition between any two levels of an \rgqbp can be split into two transitions between three levels such that the transition between the first and the second levels is query-dependent, and the transition between the second and the third levels is query-independent.
\end{lemma}

\begin{proof}
    For simplicity consider a two leveled \rgqbp $P = (Q,E,\ket{v_0}, L, \delta, F)$ of width $s$. Let $M_1 = \{m_{1,1}, \cdots, m_{1,s}\}$ and $M_2 = \{m_{2,1}, \cdots, m_{2,s}\}$ be the nodes in the first and the second levels of $P$ respectively.
    In addition, let $v_k^0$ and $v_k^1$ be the transition vectors corresponding to the $0$ and $1$ transitions of the node $m_{1,k}$.
    More formally, let the transition function $\delta$ be defined as $\delta(m_{1,k}, a, m_{2,k'}) = v_k^a[k']$.
    Recall that at any node $m_{1,k}$, the transition vectors $v_k^0$ and $v_k^1$ are related as $v_k^1 = e^{i\theta_k}v_k^0$ for all $k$.
    Now, introduce a new set of $s$ nodes $\hat{M} = \{\hat{m}_{1,1}, \cdots, \hat{m}_{1,s}\}$.
    Define a transition function $\delta'$ as $\delta'(\hat{m}_{1,k}, a, m_{2,k'}) = v_k^0[k']$ and
    $$\delta'(m_{1,k}, a, \hat{m}_{1,k'}) = \begin{cases}
        e^{i\theta_k}, & \text{if } k=k'~\text{and}~a=1\\
        1, & \text{if } k=k'~\text{and}~a=0\\
        0, & \text{else}.
    \end{cases}$$
    Notice that for any node $\hat{m}_{1,k}$ in level two, the transition vector corresponding to that node is $v^0_k$ and is independent of the value $a$ of the query index.
    
    Let $Q' = M_1 \cup \hat{M_1} \cup M_2$ and $E' = \{(v,v') : v\in M_1, v'\in \hat{M_1}, \text{ or }v\in \hat{M_1}, v;\in M_2\}$.
    Finally, define a variable assignment function $L' : Q' \rightarrow \{1,2,\cdots, n\}$ as $L'(v) = L(v)$ if $v\in Q$ and $L'(v) = 1$ if $v\notin Q$.
    Using these definitions, construct a new \rgqbp $R = (Q', E', \ket{v_0}, L', \delta', F)$. 
    To prove the theorem, it suffices to show that for any input $X\in \{0,1\}^n$ and an arbitrary initial state $\ket{v_0}$, the final states of $P$ and $R$ are identical.

    Let $\ket{v_0} = \sum_{k\in [s]} \alpha_s \ket{m_{1,k}}$ and let $X\in\{0,1\}^n$ be any arbitrary input.
    Then, the state of $P$ obtained after evolving $\ket{v_0}$ with $\delta$ will be
    \begin{align*}
        \ket{\phi_P} &= U^{O_X}_1\ket{v_0} = \sum_{k\in [s]} \alpha_k U^{O_X}_P\ket{m_{1,k}} = \sum_{k\in [s]} \alpha_k \sum_{k'\in[s]}\delta(m_{1,k}, X_{L(m_{1,k})}, m_{2,k'}) \ket{m_{2,k'}}\\
        &= \sum_{k\in [s]} \alpha_k \sum_{k'\in[s]} v_k^{X_{\tau_k}}[k'] ~\ket{m_{2,k'}} = \sum_{k'\in [s]} \bigg[\sum_{k\in [s]} \alpha_k v_k^{X_{\tau_k}}[k'] \bigg] \ket{m_{2,k'}}.
    \end{align*}
    where $\tau_k = L(m_{1,k})$, for $k\in [s]$, is the query index corresponding to the node $m_{1,k}$.
    
    In the \rgqbp $R$, the state obtained after the first transition is 
    \begin{align*}
        \ket{\psi} &= U'^{O_X}_1\ket{v_0} = \sum_{k\in [s]} \alpha_k U'^{O_X}_R\ket{m_{1,k}}
        = \sum_{k\in [s]} \alpha_k \sum_{k'\in[s]}\delta'(m_{1,k}, X_{L(m_{1,k})}, \hat{m}_{1,k'}) \ket{\hat{m}_{1,k'}} \\
        &= \sum_{k\in [s]} \alpha_k e^{i\theta_k\cdot X_{\tau_k}} \ket{\hat{m}_{1,k}}
    \end{align*}
    Then, the final state after the second transition can be given as
    \begin{align*}
        \ket{\phi_R} &= U'^{O_X}_2\ket{\psi} = \sum_{k\in [s]} \alpha_k e^{i\theta_k\cdot X_{\tau_k}} U'^{O_X}_2\ket{\hat{m}_{1,k}}
        = \sum_{k\in [s]} \alpha_k e^{i\theta_k\cdot X_{\tau_k}} \sum_{k'\in [s]} \delta'(\hat{m}_{1,k}, L(\hat{m}_{1,k}), m_{2,k'})\ket{m_{2,k'}}\\
        &= \sum_{k'\in [s]} \bigg[\sum_{k\in [s]} \alpha_k e^{i\theta_k\cdot X_{\tau_k}}  \delta'(\hat{m}_{1,k}, X_1, m_{2,k'})\bigg]\ket{m_{2,k'}} = \sum_{k'\in [s]} \bigg[\sum_{k\in [s]} \alpha_k e^{i\theta_k\cdot X_{\tau_k}}  v^0_k[k']\bigg]\ket{m_{2,k'}}
    \end{align*}
    Note that if $X_{\tau_k} = 0$, then $e^{i\theta_k X_{\tau_k}}v^0_k[k'] = v^0_k[k']$. Whereas if $X_{\tau_k} = 1$, then $e^{i\theta_k X_{\tau_k}}v^0_k[k'] = e^{i\theta_k}v^0_k[k'] = v^1_k[k']$.
    So, we have $e^{i\theta_k X_{\tau_k}}v^0_k[k'] = v^{X_{\tau_k}}_k[k']$.
    This gives us
    \begin{align*}
        \ket{\phi_R} &= \sum_{k'\in [s]} \bigg[\sum_{k\in [s]} \alpha_k e^{i\theta_k\cdot X_{\tau_k}}  v^0_k[k']\bigg]\ket{m_{2,k'}}
        = \sum_{k'\in [s]} \bigg[\sum_{k\in [s]} \alpha_k  v^{X_{\tau_k}}_k[k']\bigg]\ket{m_{2,k'}}
        = \ket{\phi_P}.
    \end{align*}
    This shows that any transition between two levels of an \rgqbp can be split into a query-dependent transition followed by a query-independent transition.
\end{proof}
While Lemma \ref{lemma:rgqbp-layer-split} deals with splitting one transition between two layers, it can be seen that this can be easily extended to an \rgqbp of any length. Thus, we obtain the following corollary.

\begin{corollary}
    \label{corr:rgqbp-alternate-layers}
    For any \rgqbp $P$ of length $l$ and width $s$, there exists an \rgqbp $Q$ of length $2l$ and width $s$ such that 
    \begin{enumerate}
        \item The levels of $Q$ alternate between a query-dependent level and a query-independent level, with the first level being query-dependent.
        \item $Q$ is equivalent to $P$, i.e, $Q$ exactly simulates $P$ and $P$ exactly simulates $Q$.
    \end{enumerate}
\end{corollary}

\begin{proof}
    Let $P$ be an \rgqbp of length $l$ and width $s$.
    We can obtain a new \rgqbp $Q$ by using Lemma~\ref{lemma:rgqbp-layer-split} and replacing each transition in $P$ with two transitions, one query-dependent, followed by a query-independent transition in $Q$.
    Since the number of transitions in the original \rgqbp $P$ is $l$, the number of transitions in $Q$ obtained as a result of replacing the transitions is $2l$.
    So, the length of $Q$ is $2l$.
    The width, however, remains the same as that of $P$.
    The equivalence of $P$ and $Q$ is a direct consequence of Lemma~\ref{lemma:rgqbp-layer-split}.
\end{proof}

\subsection*{Equivalence of r-GQBPs and quantum query circuits:}

We now show the equivalence between r-GQBPs and quantum query circuits.

\begin{theorem}
    \label{thm:rgqbp-to-circ}
     Any $q$-qubit $t$-query quantum query circuit $C$ can be exactly simulated by a $2^q$-width $t$-length \rgqbp.
\end{theorem}

We first explain the reduction at an intuitive level before proceeding with the formal proof.
Let $C$ be a query circuit that makes $t$ queries and uses $q$ qubits. Then $C$ can be decomposed into input-independent unitary gates that act on $q$ qubits and input-dependant oracle gates that, w.l.o.g. act on the first $\log n$ qubits, as $C = U_t(O_X\otimes \mathbb{I}^{\otimes{q-\log n}})U_{t-1}\cdots U_1(O_X\otimes \mathbb{I}^{\otimes{q-\log n}})U_0$.
To understand the reduction better, consider an encoding between the set of basis states $\{\ket{0}, \ket{1}, \cdots, \ket{2^q-1}\}$ and the set of states of the branching program 
$\{\ket{v_{i,0}}, \ket{v_{i,1}}, \cdots, \ket{v_{i,2^q-1}}\}$ for each level $i$.

Now, the idea is to define the start state $\ket{v_0}$, the labelling function $L$ and the transition function such that if the state of the circuit after the unitary $U_i$ is $\ket{\psi_C} = \sum_{j\in [2^q]}\alpha_{i,j}\ket{j}$, then the state of the \rgqbp should be $\ket{\psi_R} = \sum_{j\in [2^q]} \alpha_{i,j}\ket{v_{i,j}}$.
To do this, we set the initial state of $R$ as the encoded state corresponding to the circuit state $U_0\ket{0^q}$ and define a labelling function and a transition function that mimics the operation $U_i(O_X\otimes \mathbb{I}^{\otimes q-\log n})$ at the $i^{th}$ transition.
Finally, if the query circuit accepts the input on observing a state $\ket{j}$ after measuring the final state in the computational basis, then we include the node $v_{l,j}$ to the accept set $F$ of the \rgqbp $R$.
We present the formal proof below.

\newcommand{\enc}{\mathcal{E}}
\newcommand{\invenc}{\mathcal{E}^{-1}}

\begin{proof}
    Let $C$ be a $q$-qubit $t$-query quantum query circuit with access to a standard phase oracle $O_X$ that acts as $O_X\ket{i} = (-1)^{X_i}\ket{i}$.
    {\it Wlog.}\ the action of $O_X$ can be restricted to the first $\log n$ qubits of $C$.
    It is standard to represent such a query circuit as $C = \tilde{U}_t (O_X \otimes \mathbb{I}^{\otimes q-\log n})\tilde{U}_{t-1}\cdots \tilde{U}_1 (O_X\otimes \mathbb{I}^{\otimes q-\log n}) \tilde{U_0}$. Let $F_C$ be the set of the basis states, which, when obtained post the final measurement, the circuit accepts the input.

    Now, construct a \rgqbp $R = (Q, E, \ket{v_0}, L, \delta, F)$ as follows. Set $Q = \{v_{i,j} : i\in \{0,1,\cdots, t\}, j\in \{0,1,\cdots, 2^q-1\}\}$ and $E = \{(v_{i,j}, v_{i+1,j'}) : i\in \{0,1,\cdots, t-1\}, j,j' \in \{0,1,\cdots, 2^q-1\}\}$.
    For each level $i$, consider the encoding $\enc_i$ that encodes a state in a circuit to a superposition of nodes in an \rgqbp as $\enc_i\big(\sum_{j\in [2^q]} \alpha_j \ket{j}\big) = \sum_{j\in [2^q]} \alpha_j \ket{v_{i,j}}$ for any $(\alpha_0, \alpha_1, \cdots, \alpha_{2^q-1})\in \mathbb{C}^{2^q}$. Note that this encoding is a bijection; we denote its inverse as $\invenc_i$.
    Let $\ket{v_0} = \enc_0(U_0\ket{0^q})$. Set $F$ such that $F$ contain nodes of the final level corresponding to the basis states in $F_C$, i,e, $F = \{\enc_t(j) : \ket{j}\in F_C\} = \{v_{t,j} : \ket{j}\in F_C\}$.
    Now, define the transition function $\delta$ as follows:
    \vspace{-0.7cm}
    \begin{equation*}
        \delta(v_{i,j}, a, v_{p, q}) = \begin{cases}
            \tilde{U}_{i}[q,j], & \text{if $a=0$ and $p = i+1$}\\
            -\tilde{U}_{i}[q,j], & \text{if $a=1$ and $p = i+1$}\\
            0, & \text{else}
        \end{cases}
    \end{equation*}
    Finally, set the variable assignment function $L$ as $L(v_{i,j}) = 
    j_{\lvert \log n}$ where $j_{\lvert \log n}$ denotes the first $\log n$ bits of the bitwise representation of $j$.
    Notice that for any node $v_{i,j}$, the transition vector corresponding to the $0$ transition is $w^0_{v_{i,j}} = [\tilde{U}_i[0,j], \tilde{U}_i[1,j], \cdots, \tilde{U}_i[2^q-1,j]]$ and that corresponding to the $1$ transition is $w^1_{v_{i,j}} = [-\tilde{U}_i[0,j], -\tilde{U}_i[1,j], \cdots, -\tilde{U}_i[2^q-1,j]]$.
    Clearly, with $w^1_{v_{i,j}} = - w^0_{v_{i,j}}$, one of the primary conditions of an \rgqbp is satisfied.

    Due to how we have defined the set of accept states $F$, it now suffices to show that the state of the \rgqbp $R$ after the final transition is analogous to the final state of the circuit $C$.
    Mathematically, we need to show that $\ket{\psi_f} = \enc_t(\ket{\phi_f})$ where $\ket{\psi_f}$ and $\ket{\phi_f}$ are the final states of the \rgqbp and the circuit, respectively.
    At any point in a \rgqbp, $i+1^{th}$ transition will be performed at the level $i$ of the \rgqbp, i.e., when the state of the \rgqbp is a superposition of the nodes in level $i$. For instance, $1^{st}$ transition is performed at level $0$, $2^{nd}$ transition at level $1$ and so on.
    At any level $i$, the unitary corresponding to the $i+1^{th}$ transition for an input $X$ can be given as $U^{O_X}_{i+1} = \big[ (w_{v_{i,0}}^{\gamma_{i,0}})^T (w_{v_{i,1}}^{\gamma_{i,1}})^T \cdots (w_{v_{i,2^q-1}}^{\gamma_{i,1}})^T\big]_{2^q\times 2^q}$ where $\gamma_{i,j} = X_{L(v_{i,j})}$ is the value obtained on querying the index corresponding to the node $v_{i,j}$ and $(w_{v_{i,j}}^{a})^T$ denotes the column vector corresponding to the $a$-transition at the node $v_{i,j}$.
    Since, $R$ is a \rgqbp, we have, $U^{O_X}_{i+1} = \big[ (-1)^{\gamma_{i,0}}(w_{v_{i,0}}^0)^T~~(-1)^{\gamma_{i,1}}(w_{v_{i,1}}^0)^T \cdots (-1)^{\gamma_{i,2^q-1}}(w_{v_{i,2^q-1}}^0)^T\big] = \big[(w_{v_{i,0}}^0)^T~~(w_{v_{i,1}}^0)^T \cdots (w_{v_{i,2^q-1}}^0)^T\big] \cdot diag\big((-1)^{\gamma_{i,1}}, (-1)^{\gamma_{i,1}}, \cdots, (-1)^{\gamma_{i,2^q-1}}\big)$ where $diag(\cdots)$ denotes the diagonal matrix.
    However, notice that the matrix $\big[(w_{v_{i,0}}^0)^T~~(w_{v_{i,1}}^0)^T \cdots (w_{v_{i,2^q-1}}^0)^T\big]$ is exactly $\tilde{U}_{i+1}$.
    Moreover, since $(-1)^{\gamma_{i,j}} = (-1)^{X_{L(v_{i,j})}} = (-1)^{X_{\hat{j}}}$ where $\hat{j} = j_{\lvert \log n}$ for any basis state $\ket{j}$ the matrix given by $diag\big((-1)^{\gamma_{i,1}}, (-1)^{\gamma_{i,1}}, \cdots, (-1)^{\gamma_{i,2^q-1}}\big) = D\text{ (say)}$ can be given as $D = diag\big((-1)^{X_{\hat{0}}}, (-1)^{X_{\hat{1}}}, \cdots, (-1)^{X_{\hat{j}}}\big)$.
    So, in a very loose sense, the operator $D$ adds a phase to $\ket{v_{i,j}}$ depending only on the first $\log n$ bits of $j$ that serve as the index address. 
    Clearly, $D$ mimics the operator $O_X \otimes \mathbb{I}^{\otimes q - \log n}$.
    Now, if we can show that $U_{i+1}^{O_X}$ indeed mimics the action of $\tilde{U}_{i+1}(O_X \otimes \mathbb{I}^{\otimes q - \log n})$, i.e., $\enc_{i+1}\big(\tilde{U}_{i+1}(O_X \otimes \mathbb{I}^{\otimes q - \log n})\ket{\phi_i}\big) = U_{i+1}^{O_x}\ket{\psi_i}$, then we are done.
    Let $\ket{\phi_i} = \sum_{j\in [2^q]}\alpha_{i,j}\ket{j}$ be the state of the circuit after the $i^{th}$ and let $\ket{\psi_i} = \enc_i(\ket{\phi_i}) = \sum_{j\in [2^q]}\alpha_{i,j}\ket{v_{i,j}}$ be the state of the \rgqbp at level $i$.
    We have established that $\ket{\psi_0} = \ket{v_0} = \enc_0(U_0\ket{0^q}) = \enc_0(\ket{\phi_0})$.

    Now, the action of $\tilde{U}_{i+1}$ and $O_X$ on the state $\ket{\phi_i}$ can be given as
    \begin{align*}
        \tilde{U}_{i+1}(O_X\otimes \mathbb{I}^{\otimes q-\log n})\ket{\phi_i} &= \tilde{U}_{i+1}(O_X\otimes \mathbb{I}^{\otimes q-\log n}) \sum_{j\in [2^q]} \alpha_{i,j}\ket{j}\\
        &= \tilde{U}_{i+1} \sum_{j\in [2^q]}(-1)^{X_{\hat{j}}}\alpha_{i,j}\ket{j}\\
        &= \sum_{j\in [2^q]} (-1)^{X_{\hat{j}}} \alpha_{i,j}\sum_{k\in [2^q]}w_{v_i,j}^0[k]\ket{k}\\
        &= \sum_{k\in [2^q]} \Bigg[ \sum_{j\in [2^q]} (-1)^{X_{\hat{j}}}\alpha_{i,j}w_{v_i,j}^0[k] \Bigg]\ket{k}
    \end{align*}
    So, we get
    \begin{align*}
        \enc_{i+1}\big(\tilde{U}_{i+1}(O_X\otimes \mathbb{I}^{\otimes q-\log n})\ket{\phi_i}\big) &= \enc_{i+1}\Bigg(\sum_{k\in [2^q]} \Bigg[ \sum_{j\in [2^q]} (-1)^{X_{\hat{j}}}\alpha_{i,j}w_{v_i,j}^0[k] \Bigg]\ket{k}\Bigg)\\
        &= \sum_{k\in [2^q]} \Bigg[ \sum_{j\in [2^q]} (-1)^{X_{\hat{j}}}\alpha_{i,j}w_{v_i,j}^0[k] \Bigg]\ket{v_{i+1,k}}
    \end{align*}

    Next, the action of the transition defined by the unitary $U_{i+1}^{O_X}$ on the state $\ket{\psi_i}$ of the \rgqbp $R$ is
    \begin{align*}
        U_{i+1}^{O_X}\ket{\psi_i} &= U_{i+1}^{O_X} \sum_{j\in [2^q]} \alpha_{i,j} \ket{v_{i,j}} = \sum_{j\in [2^q]} \alpha_{i,j} U_{i+1}^{O_X}\ket{v_{i,j}}\\
        &= \sum_{j\in [2^q]} \alpha_{i,j} \sum_{k\in [2^q]} (-1)^{\gamma_{i,j}}w_{v_{i,j}}^0[k] \ket{v_{i+1,k}} = \sum_{k\in [2^q]} \Bigg[\sum_{j\in [2^q]} (-1)^{\gamma_{i,j}}\alpha_{i,j}w_{v_{i,j}}^0[k] \Bigg]\ket{v_{i+1,k}}\\
    \end{align*}
    Here, $(-1)^{\gamma_{i,j}} = (-1)^{X_{L(v_{i,j})}}$. But we have defined $L$ such that $L(v_{i,j}) = j_{\lvert \log n} = \hat{j}$.
    So, we have $(-1)^{\gamma_{i,j}} = (-1)^{X_{\hat{j}}}$. This gives
    \begin{align*}
        U_{i+1}^{O_X}\ket{\psi_i} &= \sum_{k\in [2^q]} \Bigg[\sum_{j\in [2^q]} (-1)^{X_{\hat{j}}}\alpha_{i,j}w_{v_{i,j}}^0[k] \Bigg]\ket{v_{i+1,k}} = \enc_{i+1}\big(\tilde{U}_{i+1}(O_X\otimes \mathbb{I}^{\otimes q-\log n})\ket{\phi_i}\big).
    \end{align*}
    Since this is true for all levels $i$, we have $\ket{\psi_t} = \enc_t(\ket{\phi_t})$.
    This shows that the \rgqbp $R$ exactly computes the quantum circuit $C$.
\end{proof}

In addition to the reduction from a quantum query circuit to an \rgqbp, we also found a reduction from an \rgqbp to a quantum circuit.
The reduction is presented as Theorem~\ref{thm:circ-to-rgqbp}.

\begin{theorem}
\label{thm:circ-to-rgqbp}
    Any $w$-width $l$-length \rgqbp $R$ can be exactly simulated by a $\log(w)+\log(n)+1$-qubit $2l$-query quantum query circuit.
\end{theorem}
\begin{proof}
    Let $R = (Q,E,\ket{v_0},L,\delta,F)$ be an ($w,l$)-\rgqbp that we would like to simulate.
    For the mechanics of the proof, we define three families of unitaries - one each corresponding to the query-independent transitions, labelling function and phase inclusion.
    \begin{itemize}
        \item First, define a family of $\log(w)$-qubit unitaries $\{\hat{U}_p\}_{p\in \{1,\cdots, l\}}$ where $U_p$ is defined as $\hat{U}_p\ket{q} = \sum_{q'\in [w]}\delta(v_{p,q}, 0, v_{p+1,q'})\ket{q'}$, for any $q\in [w]$.
        Now, let $\ket{v_0} = \sum_{q\in [w]}\alpha_{0,q}\ket{v_{0,q}}$ be the initial state of the \rgqbp $R'$.
        Then, define $\hat{U}_0$ such that $\hat{U}_0\ket{0} = \sum_{q'\in [w]}\alpha_{0,q'}\ket{q'}$.
        The unitary nature of $\hat{U}_p$ results from the fact that the transition function $\delta$ is quantumly well-behaved.
        \item Secondly, define a family of $\log w + \log n$-qubit labelling unitaries $\{\hat{LU}_p\}_{p\in \{1,\cdots, l\}}$ where the unitaries are defined as $\hat{LU}_p\ket{q}\ket{0^{\otimes \log n}} = \ket{q}\ket{L(v_{p,q})}$ where $L$ is the labelling function of the \rgqbp $R'$.
        \item Finally, define a family of $\log w +\log n + 1$-qubit phase inclusion unitaries $\{\hat{PU}_p\}_{p\in \{1,2,\cdots, l\}}$ that act as $\hat{PU}_p\ket{q}\ket{k}\ket{a} = e^{ia\theta_{p,q}}\ket{q}\ket{k}\ket{a}$.
    \end{itemize}
    Now, let $\enc_p$ be the encoding defined in the proof of Theorem~\ref{thm:circ-to-rgqbp} that encodes a state of a circuit to a state of an \rgqbp as $\enc_p\big(\sum_{q\in [w]}\alpha_q\ket{q}\big) = \sum_{q\in [w]} \alpha_q \ket{v_{p,q}}$ for any $(\alpha_0,\alpha_1,\cdots, \alpha_{s-1})\in \mathbb{C}^{w}$ for any $p$.
    Define a set $F_C$ of accept states as $F_C = \{\ket{q} : \enc_t(\ket{q})\in F\}$.
    With these families of unitaries and the set of accept states $F_C$, we can simulate \rgqbp $R'$ using circuit $C$ as in Algorithm~\ref{algo:rgqbp-to-circ}.
    We now show the correctness of the simulation algorithm.
    \begin{algorithm}[!ht]
    \caption{Simulating \rgqbp $R'$ in circuit $C$ \label{algo:rgqbp-to-circ}}
    \begin{algorithmic}[1]
        \Require Oracle $O_x$, unitaries $\{\hat{U}_p\}_p, \{\hat{LU}_p\}_p$ and $\{\hat{PU}_p\}_p$, and the set of accept states $F_C$.
        \Require Registers $R_1$, $R_2$ and $R_3$ of size $\log w, \log n$ and $1$ all set to $\ket{0}$. 
        \State Apply $\hat{U}_0$ on $R_1$.
        \For{$p$ from $1$ to $l$}
            \State Apply $\hat{LU}_p$ on $R_1$ and $R_2$.
            \State Apply $O_x$ on $R_2$ and $R_3$.
            \State Apply $\hat{PU}_p$ on $R_1, R_2$ and $R_3$.
            \State Apply $O_x$ on $R_2$ and $R_3$.
            \State Apply $\hat{LU}_p^{\dagger}$ on $R_1$ and $R_2$.
            \State Apply $\hat{U}_p$ on $R_1$.
        \EndFor
        \State Measure $R_1$ in the computational basis. If the outcome belongs in $F_C$, then accept. Else, reject.
    \end{algorithmic}
\end{algorithm}

    Let $C$ be a circuit with $\log w + \log n + 1 $ qubits. Let $\ket{\phi_{p}}$ be the state of the circuit before the $p+1^{th}$ iteration in Algorithm~\ref{algo:rgqbp-to-circ}.
    Before the first iteration, we have 
    \begin{equation*}
        \ket{\phi_0} = \hat{U}_0\ket{0}\ket{0}\ket{0}
        = \sum_{q\in [w]}\alpha_{0,q}\ket{q}\ket{0}\ket{0}
    \end{equation*}
    Clearly, $\enc_0(\ket{\phi_0}) = \ket{v_0}$.
    Now, consider the $p+1^{th}$ iteration. Let the state of the circuit $C$ before the $p+1^{th}$ iteration be $\ket{\phi_p} = \sum_{q\in [w]}\alpha_{p,q}\ket{q}\ket{0}\ket{0}$.
    Then, the evolution of the circuit during the $p+1^{th}$ iteration can be given as follows:
    \begin{align*}
        \ket{\phi_p} &\xrightarrow{\hat{LU}_p} \sum_{q\in [w]}\alpha_{p,q}\ket{q}\ket{L(v_{p,q})}\ket{0}
        \xrightarrow{O_X} \sum_{q\in [w]}\alpha_{p,q}\ket{q}\ket{L(v_{p,q})}\ket{\gamma_{p,q}}\\
        &\xrightarrow{\hat{PU}_p} \sum_{q\in [w]}e^{i\gamma_{p,q}\theta_{p,q}} \alpha_{p,q}\ket{q}\ket{L(v_{p,q})}\ket{\gamma_{p,q}}
        \xrightarrow{O_X} \sum_{q\in [w]}e^{i\gamma_{p,q}\theta_{p,q}} \alpha_{p,q}\ket{q}\ket{L(v_{p,q})}\ket{0}\\
        &\xrightarrow{\hat{LU}_p^{\dagger}} \sum_{q\in [w]}e^{i\gamma_{p,q}\theta_{p,q}} \alpha_{p,q}\ket{q}\ket{0}\ket{0} \xrightarrow{\hat{U}_p} \sum_{q\in [w]}e^{i\gamma_{p,q}\theta_{p,q}} \alpha_{p,q} 
        \hat{U}_p\ket{q}\ket{0}\ket{0}\\
        &= \sum_{q\in [w]}e^{i\gamma_{p,q}\theta_{p,q}} \alpha_{p,q} 
        \sum_{q'\in [w]} \delta(v_{p,q}, 0, v_{p+1,q'})\ket{q'}\ket{0}\ket{0}\\
        &=  \sum_{q'\in [w]} \Bigg[\sum_{q\in [w]}e^{i\gamma_{p,q}\theta_{p,q}} \alpha_{p,q} \delta(v_{p,q}, 0, v_{p+1,q'})\Bigg]\ket{q'}\ket{0}\ket{0}
    \end{align*}
    But since $R$ is an \rgqbp, we know that $\delta(v_{p,q}, 1, v_{p+1, q'}) = e^{i\theta_{p,q}}\delta(v_{p,q}, 1, v_{p+1, q'})$ from the definition of \rgqbp.
    This gives that $e^{i\gamma_{p,q}\theta_{p,q}}\delta(v_{p,q}, 0, v_{p+1, q'}) = \delta(v_{p,q}, \gamma_{p,q}, v_{p+1, q'})$. So, we have,
    \begin{align}
        \ket{\phi_{p+1}} &= \sum_{q'\in [w]} \Bigg[\sum_{q\in [w]}e^{i\gamma_{p,q}\theta_{p,q}} \alpha_{p,q} \delta(v_{p,q}, 0, v_{p+1,q'})\Bigg]\ket{q'}\ket{0}\ket{0}\notag\\
        &= \sum_{q'\in [w]} \Bigg[\sum_{q\in [w]}\alpha_{p,q} \delta(v_{p,q}, \gamma_{p,q}, v_{p+1,q'})\Bigg]\ket{q'}\ket{0}\ket{0}\label{eqn:circ-p+1}
    \end{align}

    Now, consider the $p+1^{th}$ transition in \rgqbp $R'$. Let $\ket{\psi_{p}} = \sum_{q\in [w]}\alpha_{p,q}\ket{v_{p, q}}$ be the state of $R$ just before the $p+1^{th}$ transition. Then, the state of $R$ after the $p+1^{th}$ transition can be given as 
    \begin{align}
        U^{O_X}_{p+1}\ket{\psi_{p}} &= U^{O_X}_{p+1}\sum_{q\in [w]} \alpha_{p, q}\ket{v_{p, q}}
        = \sum_{q\in [w]} \alpha_{p, q} ~U^{O_X}_{p+1}\ket{v_{p, q}}\notag\\
        &= \sum_{q\in [w]} \alpha_{p, q} \sum_{q'\in [w]}\delta(v_{p,q}, \gamma_{p,q}, v_{p+1,q'})\ket{v_{p+1, q}}\notag\\
        &= \sum_{q'\in [w]} \Bigg[\sum_{q\in [w]} \alpha_{p, q} ~\delta(v_{p,q}, \gamma_{p,q}, v_{p+1,q'})\Bigg]\ket{v_{p+1, q}} = \ket{\psi_{p+1}}\label{eqn:rgqbp-p+1}
    \end{align}
    where $\gamma_{p,q} = X_{L(v_{p,q})}$ is the value obtained on querying the index corresponding to the node $v_{p,q}$.
    If we call the state of the first register of Equation~\ref{eqn:circ-p+1} as $\ket{\phi'_{p+1}}$, then comparing it with Equation~\ref{eqn:rgqbp-p+1}, we can see that 
    \begin{equation*}
        \enc_{p+1}\big(\ket{\phi'_{p+1}}\big) = \ket{\psi_{p+1}}.
    \end{equation*}
    This is true for any $p\in \{0,1,\cdots, l-1\}$. So, we have $\enc_l\big(\ket{\phi'_l}\big) = \ket{\psi_l}$.
    Finally, on measuring $\ket{\phi'_l}$, the circuit accepts the input if and only if the measurement outcome, say $\ket{m}\in F_C$. This implies that the circuit accepts the input if and only if $\enc_l(\ket{m}) = \ket{v_{l,m}}\in F$.
    Hence, the circuit $C$ accepts an input if and only if the \rgqbp $R$ accepts the input. Moreover, since the amplitude of $\ket{m}$ in the circuit $C$ is the same as the amplitude of $\enc_l(\ket{m}) = \ket{v_{l,m}}$, $C$ accepts with the same probability as that of $R$, and this simulates $R$ exactly.  
\end{proof}
\section{Query-Space Lower Bound for Unordered Search}
\label{sec:search-lowerbound}

The unordered search problem was one of the earliest problems for which a query lower bound was proved and one of the first problems to admit a quantum speedup. 
We follow suit and present a query-space lower bound for the unordered search in the \rgqbp mode by showing a lower bound on a related problem, which we call the Promise-OR problem. 
We define the Promise-OR problem below.
\begin{problem}[Promise-OR Problem]
    Given access to an $n$-bit Boolean string $X$ with a promise that at most one of the $n$ positions has a $1$, decide if $X$ has a $1$. 
\end{problem}

\begin{figure}[b]
    \centering
    \includegraphics[width=0.8\textwidth]{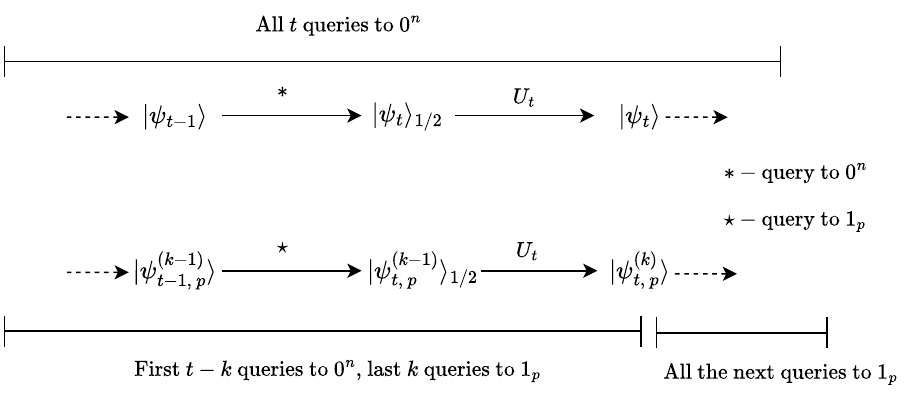}
    \caption{The state transitions using the hybrid argument.}
    \label{fig:state-transition}
\end{figure}

\lowerboundor*

\begin{proof}
   To prove this theorem, we follow the lines of the well-known hybrid argument technique for lower bounds.    
    Let $R$ be an ($s,L$)-\rgqbp that solves the Promise-OR problem, and let $R'$ be the equivalent ($s, 2L$)-\rgqbp obtained using Corollary~\ref{corr:rgqbp-alternate-layers}. We will analyse two evolutions of $R'$. The first is on the input $0^n$. The second will be a hybrid evolution, where some initial transitions will be using $0^n$ as the input string, and the rest of the transitions will be using $0^{p-1}10^{n-p}$, denoted $1_p$, as the input string. The intuition would be to show that the corresponding states after each transition in these two evolutions have a bounded difference, and unless there are lots of transitions, the final states would not be sufficiently distinguishable for input strings $0^n$ and $1_p$. We require a few notations to handle the technicalities.

    Let $U_t$ denote the unitary corresponding to the $t^{th}$ query-independent transition and $V_t^{X}$ denote the unitary corresponding to the $t^{th}$ query-dependent transition of $R'$ when then input is $X$.
    Let $\ket{\psi_t}_{1/2}$ and $\ket{\psi_t}$ be the states of the branching program after the $t^{th}$ query-independent and $t^{th}$ query-independent transitions, respectively, when querying the input $0^n$.
    Similarly, let $\ket{\psi_{t,p}^{(k)}}_{1/2}$ and $\ket{\psi_{t,p}^{(k)}}$ be the states of the branching program after the $t^{th}$ query-dependent and $t^{th}$ query-independent transitions, respectively, when the first $t-k$ queries are made to $0^n$ and the last $k$ queries are made to $0^{p-1}10^{n-p}$, denoted $1_p$; these denote the states in a hybrid run of $R'$.
    Note that unitary $U_t$ corresponding to the $t^{th}$ query-independent transition is such that $\ket{\psi_t} = U_t\ket{\psi_t}_{1/2}$ and $\ket{\psi_{t,p}^{(k)}} = U_t\ket{\psi_{t,p}^{(k)}}_{1/2}$.

    The objective is to bound the difference between the states $\ket{\psi_L}$ (final state of the non-hybrid evolution, on input $0^n$) and $\ket{\psi_{L,p}^{(L)}}$ (final state of the hybrid evolution, but all queries are to input $1_p$).
    
    Suppose we represent $\ket{\psi_t}_{1/2}$ as $\displaystyle \ket{\psi_t}_{1/2} = \sum_{j\in [s]}\alpha_{j,t}\ket{j}$. Let $Q_{p,t}$ be the set of nodes that query the index $p$ after the $t^{th}$ query-dependent transition.
    Then we can write 
    \begin{equation*}
        \ket{\psi_{t,p}^{(1)}}_{1/2} = \sum_{j\in [s] : j\notin Q_{p,t}}\alpha_{j,t}\ket{j} + \sum_{j\in Q_{p,t}} e^{i\theta_{j,t}}\alpha_{j,t}\ket{j}.
    \end{equation*}

    This allows us to compare two evolutions in which only the $t$-th query was made to different strings, $0^n$ and $1_p$, and all the earlier ones being made to $0^n$ for both.
    \begin{align*}
        \norm{\ket{\psi_t} - \ket{\psi_{t,p}^{(1)}}} &= \norm{U_t \ket{\psi_t}_{1/2} - U_t\ket{\psi_{t,p}^{(1)}}_{1/2}} = \norm{\ket{\psi_t}_{1/2} - \ket{\psi_{t,p}^{(1)}}_{1/2}}\\
        &= \norm{\Bigg[\sum_{j\in [s]}\alpha_{j,t}\ket{j}\Bigg] - \Bigg[\sum_{j\in [s] : j\notin Q_{p,t}}\alpha_{j,t}\ket{j} + \sum_{j\in Q_{p,t}} e^{i\theta_{j,t}}\alpha_{j,t}\ket{j}\Bigg]}\\
        &= \norm{\sum_{j\in Q_{p,t}} (1-e^{i\theta_{j,t}})\alpha_{j,t}\ket{j}} \le 2\norm{\sum_{j\in Q_{p,t}}\alpha_{j,t}\ket{j}}.
    \end{align*}

    \noindent Now we compare the final states of the evolutions in which only the final two queries differ.
    \begin{align*}
        \ket{\psi_L} - \ket{\psi_{L,p}^{(2)}} &= \ket{\psi_L} - U_L V_L^{1_p}\ket{\psi_{L-1}} + U_L V_L^{1_p}\ket{\psi_{L-1}} - \ket{\psi_{L,p}^{(2)}}\\
        &= \ket{\psi_L} - U_L V_L^{1_p}\ket{\psi_{L-1}} + U_L V_L^{1_p}\ket{\psi_{L-1}} -  U_L V_L^{1_p}\ket{\psi_{L-1,p}^{(1)}}\\
        &= \Big[\ket{\psi_L}-\ket{\psi_{L,p}^{(1)}}\Big] + U_L V_L^{(1_p)} \Big[\ket{\psi_{L-1}} - \ket{\psi_{L-1,p}}\Big].
    \end{align*}
    We can extend the above using induction and telescoping to show that 
    \begin{align*}
        \ket{\psi_L} - \ket{\psi_{L,p}^{(L)}} &= \sum_{t=L}^0 ~\Bigg(\prod_{l=t+1}^L U_l V_l^{(1_p)}\Bigg)~ \Big[\ket{\psi_{t}} - \ket{\psi_{t,p}}\Big].
    \end{align*}
    We can bound their distance as well.
    \begin{align*}
        \norm{\ket{\psi_L} - \ket{\psi_{L,p}^{(L)}}} &= \norm{\sum_{t=L}^0 ~\Bigg(\prod_{l=t+1}^L U_l V_l^{(1_p)}\Bigg)~ \Big[\ket{\psi_{t}} - \ket{\psi_{t,p}}\Big]}\\
        &\le \sum_{t=L}^0 \norm{\Bigg(\prod_{l=t+1}^L U_l V_l^{(1_p)}\Bigg)~ \Big[\ket{\psi_{t}} - \ket{\psi_{t,p}}\Big]}\\
        &= \sum_{t=L}^0 \norm{\prod_{l=t+1}^L U_l V_l^{(1_p)}} ~\norm{\Big[\ket{\psi_{t}} - \ket{\psi_{t,p}}\Big]}\\
        &= \sum_{t=0}^L \norm{\Big[\ket{\psi_{t}} - \ket{\psi_{t,p}}\Big]} \le \sum_{t=0}^L 2 \norm{\sum_{j\in Q_{p,t}} \alpha_{j,t}\ket{j}} \le 2 \sum_{t=0}^L \sum_{j\in Q_{p,t}}\abs{\alpha_{j,t}}
    \end{align*}
    
    Then, taking an expectation over a uniform distribution on $p$, we get
    \begin{align*}
        \expec_p \bigg[\norm{\ket{\psi_L} - \ket{\psi_{L,p}^{(L)}}}\bigg] &\le \expec_p\bigg[2 \sum_{t=0}^L \sum_{j\in Q_{p,t}}\abs{\alpha_{j,t}}\bigg]= \frac{2}{n} \sum_{p=1}^n\sum_{t=0}^L \sum_{j\in Q_{p,t}}\abs{\alpha_{j,t}}= \frac{2}{n}\sum_{t=0}^L \sum_{p=1}^n \sum_{j\in Q_{p,t}}\abs{\alpha_{j,t}}.
    \end{align*}
    Since one node queries exactly one index, each node will belong to $Q_{p,t}$ for some $p$ in each $t$. So, we have $\sum_{p=1}^n \sum_{j\in Q_{p,t}}\norm{\alpha_{j,t}} = \sum_{j\in [s]} \abs{\alpha_{j,t}}$. This gives,
    \begin{align*}
        \expec_p \bigg[\norm{\ket{\psi_L} - \ket{\psi_{L,p}^{(L)}}}\bigg] &= \frac{2}{n}\sum_{t=0}^L \sum_{j\in [s]} \abs{\alpha_{j,t}}.
    \end{align*}
    To maximize the value of $\sum_{j\in [s]} \abs{\alpha_{j,t}}$ for any fixed $t$ conditioned on the fact that $\sum_{j\in [s]} \abs{\alpha_{j,t}}^2 = 1$, we set $\abs{\alpha_{j,t}} = \frac{1}{\sqrt{s}}$ for all $j$.
    So, we get
    \begin{align}
        \expec_p \bigg[\norm{\ket{\psi_L} - \ket{\psi_{L,p}^{(L)}}}\bigg] &= \frac{2}{n}\sum_{t=0}^L \sum_{j\in [s]} \abs{\alpha_{j,t}} \le \frac{2}{n}\sum_{t=0}^L \sum_{j\in [s]} \frac{1}{\sqrt{s}} = \frac{2(L+1)\sqrt{s}}{n} \label{proof:expec_1}. 
    \end{align}

    But to be able to successfully distinguish between the inputs $X=0^n$ and $X=1^p$ for any $p$, we need that $\norm{\ket{\psi_L} - \ket{\psi_{L,p}^{(L)}}} = \Omega(1)$.
    Again, taking an expectation over the uniform distribution of $p$, we get
    \begin{align}
        \expec_p \bigg[\norm{\ket{\psi_L} - \ket{\psi_{L,p}^{(L)}}}\bigg] &= \Omega(1) \label{proof:expec_2}.
    \end{align}
    Comparing (\ref{proof:expec_1}) and (\ref{proof:expec_2}), we get $\tfrac{2(L+1)\sqrt{s}}{n} = \Omega(1)$, or $L = \Omega\bigg(\frac{n}{\sqrt{s}}\bigg)$.
\end{proof} 
\section{Query-Space Lower Bound for ($k,k+\delta$) Hamming Decision Problem}
\label{sec:k-decision-lowerbound}

\newcommand{\khamprob}{($k,k+\delta$) Hamming decision problem\xspace}

Following the lower bound obtained for the Promise-OR problem, we extend a similar query-space lower bound for the ($k,k+1$) Hamming decision problem.
We first formally define the \khamprob.
\begin{problem}[($k,k+\delta$) Hamming Decision Problem]
    Given access to an $n$-bit Boolean input $X$ with the promise that the Hamming weight of $X$ is either $k$ or $k+\delta$ for some $\delta< k$, decide if the Hamming weight of $X$ is $k$ or $k+\delta$.
\end{problem}

The \khamprob is an interesting problem to consider. For any symmetric function, there exists some fixed $k$ such that ($k,k+1$) Hamming decision problem can be reduced to that particular symmetric function.
For instance, when $k=n/2$, then ($k,k+1$) Hamming decision problem can be reduced to problems like Parity and Majority and for any fixed $k$, ($k,k+1$) Hamming decision problem can be reduced to the threshold-$k$ function.
So, a lower bound on the \khamprob essentially provides a lower bound for all the symmetric functions.

\hammingdecision*

The proof of Theorem~\ref{thm:lowerbound-khamprob} follows a pattern similar to the proof of Theorem~\ref{thm:lowerbound-or}.
The main difference here is that in the proof of the earlier, query-dependent transitions did not affect the state of the branching program when the input is $X=0^n$.
However, here, the query-dependent layer could potentially change the state of the branching program irrespective of the input.
The structure of the proof is still the same as that of the earlier.
We first use hybrid arguments to bound the norm of the difference of the states obtained when input to the \rgqbp is a yes instance and a no instance.
Then, fixing a yes instance, we take the expectation value of this norm over a uniform distribution on a set of hard no instances.
Finally, we note that the expectation value has to be $\Omega(1)$ in order to differentiate between the yes instance and the no instance with a constant probability.
While the last step is easy, the first two steps are slightly complicated.
We present the proof of Theorem~\ref{thm:lowerbound-khamprob} below.

\begin{proof}[Proof of Theorem~\ref{thm:lowerbound-khamprob}]
    
    Let $x$ correspond to an input such that $|x| = k$ and $y$ correspond to an input such that $|y| = k+1$.
    Let $R$ be a ($s,L$)-\rgqbp that solves the \khamprob and let $R'$ be its dual ($s,2L$)-\rgqbp obtained from Corollary~\ref{corr:rgqbp-alternate-layers}.
    We introduce a few notations required for the proof similar to that in the proof of Theorem~\ref{thm:lowerbound-or}.
    Let $\ket{\psi_{t,w}}_{1/2}$ and $\ket{\psi_{t,w}}$ be the state of the branching program after the $t^{th}$ query-dependent and $t^{th}$ query-independent transitions, respectively, when the input is $w\in \{x,y\}$.
    Let $\ket{\psi_{t,y}^{(r)}}_{1/2}$ and $\ket{\psi_{t,y}^{(r)}}$ denote the states of the branching program after the $t^{th}$ query-dependent and 
    $t^{th}$ query-independent transitions, respectively, where the first $t-k$ query-dependent transitions were made to the input $x$ and the last $k$ query-dependent transitions were made to the input $y$.
    Notice that $\ket{\psi_{t,y}^{(t)}}_{1/2} = \ket{\psi_{t,y}}_{1/2}$ and $\ket{\psi_{t,y}^{(t)}} = \ket{\psi_{t,y}}$, i.e, in the hybrid run, when the last $t$ query-dependent transitions out of a total $t$ transitions are made to the input $y$, then it is the same as evaluating the \rgqbp $R'$ on the input $y$.

    Our ultimate objective here is to upper bound the difference between the states $\ket{\psi_{t,L}}$ and $\ket{\psi_{t,L}}$.
    Now, let the state just before the $t^{th}$ query dependent level when the input is $x$ be denoted as
    \begin{equation*}
        \ket{\psi_{t-1,x}} = \sum_{j\in [s]}\alpha_{j,t}\ket{j}.
    \end{equation*}
    The state obtained when the next query is made to $y$ is
    \begin{equation*}
        \ket{\psi_{t,y}^{(1)}}_{1/2} = \sum_{j\in [s]: j\notin Q_{y,t}} \alpha_{j,t} \ket{j} + \sum_{j\in [s]: j\in Q_{y,t}} e^{i\theta_{j,t}}\alpha_{j,t} \ket{j},
    \end{equation*}
    where $Q_{y,t}$ is the set of nodes at the $t^{th}$ query-dependent level that queries an index $p$ such that $y_p=1$.
    Similarly, the state of the branching program, when the $t^{th}$ query was made to $x$, will be
    \begin{equation*}
        \ket{\psi_{t,x}}_{1/2} = \sum_{j\in [s]: j\notin Q_{x,t}} \alpha_{j,t} \ket{j} + \sum_{j\in [s]: j\in Q_{x,t}} e^{i\theta_{j,t}}\alpha_{j,t} \ket{j},
    \end{equation*}
    Using these two expressions, we can bound the difference between these states as
    \begin{align*}
        \norm{\ket{\psi_{t,x}} - \ket{\psi_{t,y}^{(1)}}} &= \norm{U_t\ket{\psi_{t,x}}_{1/2} - U_t\ket{\psi_{t,y}^{(1)}}_{1/2}} = \norm{U_t}\norm{\ket{\psi_{t,x}}_{1/2} - \ket{\psi_{t,y}^{(1)}}_{1/2}}\\
        &= \norm{\ket{\psi_{t,x}}_{1/2} - \ket{\psi_{t,y}^{(1)}}_{1/2}}\\
        &= \Bigg\lvert\Bigg\lvert\Bigg[\sum_{j\in [s]: j\notin Q_{x,t}} \alpha_{j,t} \ket{j} + \sum_{j\in [s]: j\in Q_{x,t}} e^{i\theta_{j,t}}\alpha_{j,t} \ket{j}\Bigg] \\
        &\hspace{5em}- \Bigg[\sum_{j\in [s]: j\notin Q_{y,t}} \alpha_{j,t} \ket{j} + \sum_{j\in [s]: j\in Q_{y,t}} e^{i\theta_{j,t}}\alpha_{j,t} \ket{j}\Bigg]\Bigg\rvert\Bigg\rvert\\
        &= \norm{\sum_{j\in Q_{x,t}~\&~j\notin Q_{y,t}} (e^{i\theta_{j,t}} - 1)\alpha_{j,t}\ket{j} + \sum_{j\notin Q_{x,t}~\&~j\in Q_{y,t}} (1-e^{i\theta_{j,t}})\alpha_{j,t}\ket{j}}\\
        &\le \sum_{j\in \Delta(x,y,t)} \norm{(1- e^{i\theta_{j,t}})\alpha_{j,t}\ket{j}}\\
        &\le 2 \sum_{j\in \Delta(x,y,t)} \norm{\alpha_{j,t}\ket{j}} = 2 \sum_{j\in \Delta(x,y,t)} \abs{\alpha_{j,t}},
    \end{align*}
    where $\Delta(x,y,t)$ is the set of all nodes that query an index $p$ such that $x_p\neq y_p$  in the $t^{th}$ query-dependent level.

    Similar to the proof of Theorem~\ref{thm:lowerbound-or}, it is easy to show that
    \begin{equation*}
        \norm{\ket{\psi_{L,x}} - \ket{\psi_{L,y}}} = \norm{\ket{\psi_{L,x}} - \ket{\psi_{L,y}^{(L)}}} \le 2\sum_{t=0}^L \sum_{j\in \Delta(x,y,t)} \abs{\alpha_{j,t}}.
    \end{equation*}
    We consider two cases --- $k\le n/2$ and $k > n/2$.

    \textbf{Case 1: $k\le n/2$.} Now, fix the input $x$ such that $k=|x|\le n/2$. Let $\mathscr{R}$ be a set of inputs $y$ such that $x_i = y_i$ whenever $x_i=1$ and out of the remaining $n-k$ bit positions where $x_i = 0$, $y$ differs from $x$ at exactly $\delta$ positions.
    Then, the number of elements in $\mathscr{R}$ will be exactly $\binom{n-k}{\delta}$.
    Taking an expectation over a uniform distribution of $y$ in $\mathscr{R}$, we get
    \begin{align*}
        \expec_\mathscr{R} \Big[ \norm{\ket{\psi_{L,x}} - \ket{\psi_{L,y}}}\Big]
        &\le  \expec_\mathscr{R}\Big[ 2\sum_{t=0}^L \sum_{j\in \Delta(x,y,t)} \abs{\alpha_{j,t}} \Big]\\
        &= \frac{2}{\binom{n-k}{\delta}} \sum_{y\in \mathscr{R}} \sum_{t=0}^L \sum_{j\in \Delta(x,y,t)} \abs{\alpha_{j,t}} = \frac{2}{\binom{n-k}{\delta}} \sum_{t=0}^L \sum_{y\in \mathscr{R}} \sum_{j\in \Delta(x,y,t)} \abs{\alpha_{j,t}}
    \end{align*}
    Notice the last two summations of the last equation. Take any index $i\in [n]$. Then, the number of $y\in \mathscr{R}$, such that $x_i\neq y_i$ is either $0$ (for $i$ such that $x_i=1$) or exactly $\binom{n-k-1}{\delta-1}$.
    So, for any $j\in [s]$, the number of $y\in \mathscr{R}$ such that $j\in \Delta(x,y,t)$ is either $0$ or $\binom{n-k-1}{\delta-1}$.
    This is because of the fact that $x$ and $y$ differ exactly at $\delta$ positions, and these are positions where $x_i=0$.
    So, we get
    \begin{equation*}
        \sum_{y\in\mathscr{R}}\sum_{j\in \Delta(x,y,t)} \abs{\alpha_{j,t}} 
        ~\le~ \binom{n-k-1}{\delta-1} \sum_{{j\in [s] :\atop j\in \Delta(x,y,t)}\atop\text{~for some }y\in \mathscr{R}}\abs{\alpha_{j,t}} 
        ~\le~ \binom{n-k-1}{\delta-1} \sum_{j\in [s]} \abs{\alpha_{j,t}}
    \end{equation*}
    Using this in the expression of expectation, we get
    \begin{align*}
        \expec_\mathscr{R} \Big[ \norm{\ket{\psi_{L,x}} - \ket{\psi_{L,y}}}\Big] &= \frac{2}{\binom{n-k}{\delta}} \sum_{t=0}^L \sum_{y\in \mathscr{R}} \sum_{j\in \Delta(x,y,t)} \abs{\alpha_{j,t}} \\
        &\le \frac{2\binom{n-k-1}{\delta-1}}{\binom{n-k}{\delta}} \sum_{t=0}^L \sum_{j\in [s]}\abs{\alpha_{j,t}}
        ~~\le  \frac{2\delta}{n-k} \sum_{t=0}^L \sum_{j\in [s]}\abs{\alpha_{j,t}}
    \end{align*}
    Since, $\sum_{j\in [s]} \abs{\alpha_{j,t}}^2 = 1$, setting $\abs{\alpha_{j,t}} = \frac{1}{\sqrt{s}}$ maximizes $\sum_{j\in [s]}\abs{\alpha_{j,t}}$.
    So, we get,
    \begin{equation}
        \label{eqn:expect_khamprob1}
        \expec_\mathscr{R} \Big[ \norm{\ket{\psi_{L,x}} - \ket{\psi_{L,y}}}\Big] ~\le~ \frac{2\delta}{n-k} \sum_{t=0}^L \sum_{j\in [s]}\abs{\alpha_{j,t}} \le \frac{2\delta}{n-k} \sum_{t=0}^L \sqrt{s} = \frac{2(L+1)\delta\sqrt{s}}{n-k}.
    \end{equation}
    For any two inputs $x, y$ such that $f(x)\neq f(y)$, we need $\norm{\ket{\psi_{L,x}} - \ket{\psi_{L,y}}} = \Omega(1)$ to be able to distinguish them with some constant probability.
    With a fixed $x$, taking an expectation over all inputs $y$ in $\mathscr{R}$, we have $\expec_\mathscr{R} \Big[ \norm{\ket{\psi_{L,x}} - \ket{\psi_{L,y}}}\Big] = \Omega(1)$. Comparing this with the expression in Equation~\ref{eqn:expect_khamprob1}, we get
    \begin{equation*}
        \frac{2(L+1)\delta\sqrt{s}}{n-k} = \Omega(1) ~~~\text{or, }~~~ L = \Omega\Big(\frac{n-k}{\delta\sqrt{s}}\Big)
    \end{equation*}
    As $k\le n/2$, we have $n-k = \Omega(n)$ which gives $L = \Omega\Big(\frac{n}{\delta\sqrt{s}}\Big)$.

    \textbf{Case 2: $k>n/2$.} Fix an input $y$ such that $|y|=k+\delta$. Let $\mathscr{R}$ be the set of inputs $x$ such that $x_i=0$ whenever $y_i=0$ and $x$ differs from $y$ at exactly $\delta$ position out of the remaining $k+\delta$ positions.
    It is easy to see that the number of elements in $\mathscr{R}$ is $\binom{k+\delta}{k}$. Similar to Case 1, by taking an expectation over a uniform superposition of $x$ in $\mathscr{R}$ and using similar arguments, we can show that
    \begin{align*}
        \expec_\mathscr{R} \Big[ \norm{\ket{\psi_{L,x}} - \ket{\psi_{L,y}}}\Big] &\le \frac{2}{\binom{k+\delta}{k}} \sum_{t=0}^L \sum_{x\in \mathscr{R}} \sum_{j\in \Delta(x,y,t)} \abs{\alpha_{j,t}}
    \end{align*}
    Now, for any index $i\in [n]$, the number of inputs $x\in \mathscr{R}$ such that $x_i\neq y_i$ is either $0$ or exactly $\binom{k+\delta-1}{k}$.
    So, for any $j\in [s]$, the number of inputs $x\in \mathscr{R}$ such that $j\in \Delta(x,y,t)$ is either $0$ or $\binom{k+\delta-1}{k}$.
    This gives us
    \begin{equation*}
        \sum_{x\in\mathscr{R}}\sum_{j\in \Delta(x,y,t)} \abs{\alpha_{j,t}} 
        ~\le~ \binom{k+\delta-1}{k} \sum_{{j\in [s] :\atop j\in \Delta(x,y,t)}\atop\text{~for some }x\in \mathscr{R}}\abs{\alpha_{j,t}} 
        ~\le~ \binom{k+\delta-1}{k} \sum_{j\in [s]} \abs{\alpha_{j,t}}.
    \end{equation*}
    Using this, we refine the bound on the expectation as
    \begin{align}
        \expec_\mathscr{R} \Big[ \norm{\ket{\psi_{L,x}} - \ket{\psi_{L,y}}}\Big] 
        &\le \frac{2\binom{k+\delta-1}{k}}{\binom{k+\delta}{k}} \sum_{t=0}^L \sum_{j\in [s]} \abs{\alpha_{j,t}}
        \le \frac{2\delta}{k+\delta} \sum_{t=0}^L \sum_{j\in [s]} \abs{\alpha_{j,t}}\notag\\
        &\le \frac{2(L+1)\delta\sqrt{s}}{k}\label{eqn:expec-khamprob2}
    \end{align}
    For any two inputs $x, y$ such that $f(x)\neq f(y)$, we need $\norm{\ket{\psi_{L,x}} - \ket{\psi_{L,y}}} = \Omega(1)$ to be able to distinguish them with some constant probability.
    With a fixed $y$, taking an expectation over all inputs $x$ in $\mathscr{R}$, we have $\expec_\mathscr{R} \Big[ \norm{\ket{\psi_{L,x}} - \ket{\psi_{L,y}}}\Big] = \Omega(1)$.
    Comparing this with the expression in Equation~\ref{eqn:expec-khamprob2}, we get,
    \begin{equation*}
        \frac{2(L+1)\delta\sqrt{s}}{k} = \Omega(1) ~~~\text{or, }~~~ L = \Omega\Big(\frac{k}{\delta\sqrt{s}}\Big).
    \end{equation*}
    Since, $k>n/2$, we get $L = \Omega\Big(\frac{n}{\delta\sqrt{s}}\Big)$.
\end{proof}

Setting $\delta = O(1)$ gives rise to the following corollary.
\begin{corollary}
    \label{corr:lowerbound-gen-delta=1}
    Any \rgqbp of length $L$ and width $s$ that solves the ($k,k+1$) Hamming Decision problem has to satisfy $L = \Omega\Big(\frac{n}{\sqrt{s}}\Big)$.
\end{corollary}

From this result, we can obtain a query-space lower bound for the class of all total symmetric functions.

\lowerboundsymmgqbp*

\begin{proof}
    Let $f : \{0,1\}^n \rightarrow \{0,1\}$ be a non-constant symmetric function.
    Then, there should exist some $k$ such that for all inputs $x$ of Hamming weight $k$, $f(x) = a$ for some $a\in \{0,1\}$ and all inputs $y$ of Hamming weight $k+1$, $f(y) = \overline{a}$.
    Now, there is a direct reduction from the partial function corresponding to the \khamprob to $f$ that is $g(x)=k \iff f(x)=a$ and $g(x)=k+1 \iff f(x)=\overline{a}$.
    This reduction ensures that any lower bound for solving the \khamprob also serves as a lower bound for computing $f$.
    Hence, we obtain the lower bound for computing $f$ from Theorem~\ref{thm:lowerbound-khamprob}.
\end{proof}

We can now apply this result to functions such as AND, Parity, and Majority and obtain the tradeoff of $L^2 s = \Omega(n^2)$. This tradeoff is tight for OR and Parity since Bera et al.\ constructed \gqbp with the same tradeoff in their work (see \cref{fig:gqbp-parity}) ~\cite{bera2023generalized}.

Since we have tight conversion algorithms between \gqbp and the quantum-query circuit, a bound on one model implies a bound on the other. We apply this idea on \cref{thm:lowerbound-gqbp-symm-function} to obtain the following lower bound on quantum query circuits.

\symmetricfunction*

It is known that computing any non-constant symmetric Boolean function on $n$-bit inputs requires $\Omega(\sqrt{n})$ queries to the input~\cite{beals2001quantum}; our result provides an alternative proof of a weaker form of this result where there is a restriction on the space.

\section{Conclusions}
\label{sec:conclusions}

In this work, we introduced a restricted version of the recently introduced \gqbp model, called \rgqbp.
Along with a few properties of r-GQBPs, we showed that any $q$-qubit 
$t$-query quantum query circuit can be reduced to a $2^q$-width $t$-length \rgqbp.
This reduction serves as a motivation to investigate the query-space lower bounds in the \rgqbp model since any lower bound on this model translates to a lower bound on the query circuit model without any overhead.
In addition, we also stated a reduction from the \rgqbp model to the query circuit model, although with some overhead.

We also present a series of lower bounds starting with a query-space lower bound of $L^2s = \Omega(n^2)$ for the unordered search in the \rgqbp model. Moreover, we extend this result to the \khamprob and proved a query-space lower bound of $L^2s = \Omega(n^2/\delta)$.
By a simple reduction from the ($k, k+1$) Hamming decision problem, we show that any \rgqbp that approximately computes any non-constant symmetric Boolean function has to satisfy $L^2s = \Omega(n^2)$.
While this lower bound is tight for functions such as $OR_n$, $AND_n$, and $Parity_n$, the tightness is unclear for functions like $Majority_n$ and $Threshold^k_n$.
With advice from the known results, we conjecture that any \rgqbp that approximately computes $Majority_n$ has to satisfy $Ls = \Omega(n^2)$ while the bound for $Threshold^k_n$ will be $Ls = \Omega(nk)$.

All the lower-bound results presented in this paper implicitly assume $s\in o(n^2)$; otherwise, we get trivial lower bounds. To obtain non-trivial lower bounds for query circuits, we need to drop this assumption and, instead, investigate r-GQBPs of larger width.

While the lower bounds for query circuits seem not so significant, we strongly believe that the r-GQBP lower bounds will lead to significant time-space lower bounds for computing functions in quantum circuits and Turing machines and is an open direction worth exploring.



\bibliography{refs}




\end{document}